%% file: main.tex
\documentclass[11pt]{article}

\title{Distributed Dominating Set With Optimal Rounds and Message Size in Bounded Arboricity Graphs}

\author{Sharareh Alipour\thanks{Tehran Institute for Advanced Studies (TeIAS), Khatam University, \texttt{sharareh.alipour@gmail.com}}
\and Ermiya Farokhnejad\thanks{University of Warwick, \texttt{Ermiya.Farokhnejad@warwick.ac.uk}}}

\date{}

\input{commands}

\begin{document}

\maketitle

\begin{abstract}
    We study the distributed minimum dominating set problem on graphs of arboricity $\alpha$.
    Dory, Ghaffari, and Ilchi [PODC'22] showed that any algorithm achieving a constant or poly-logarithmic approximation factor needs at least $\Omega(\log\Delta/\log\log\Delta)$ rounds in graphs of maximum degree $\Delta$ and arboricity $\alpha$, even when $\alpha=2$ and even when the message sizes are unbounded.
    Although there is a variety of algorithms with a near-optimal round complexity of $O(\log\Delta)$, it is natural to ask: \textbf{What is the best approximation factor in the optimal round complexity of $O(\log\Delta/\log\log\Delta)$?}
    
    We make progress in answering this question by describing a deterministic algorithm that obtains a $O\left( \alpha \log \Delta / \log\log \Delta \right)$ approximation \textbf{without prior knowledge} of $\alpha$ with \textbf{optimal round complexity} of $O\left( \log \Delta / \log\log \Delta  \right)$ and \textbf{optimal message size} of $1$ bit per round.
    
    Among all of the previous results, the only algorithm that achieves the optimal round complexity of $O\left( \log \Delta / \log\log \Delta  \right)$ \textbf{without} prior knowledge of $\alpha$ is due to Lenzen and Wattenhofer [DISC'10] that obtains a $O(\alpha \log^{1+\varepsilon}\Delta / (\varepsilon\log\log \Delta))$ approximation in $O(\log\Delta/(\varepsilon\log\log\Delta))$ rounds and $O(\log(\varepsilon^{-1}\log\Delta))$ message size.
    Our algorithm simplifies and improves upon this result.
    The only downside of our algorithm compared to the algorithm of Lenzen and Wattenhofer is that it needs prior knowledge of $\Delta$.

    The previous state-of-the-art algorithm by Dory, Ghaffari, and Ilchi [PODC'22] has a dependency on $\log n$ in the round complexity for \textbf{unknown} $\alpha$, which is far from optimal.
\end{abstract}

\newpage

\input{introduction}

\input{algorithm}

\bibliographystyle{alpha}
\bibliography{references.bib}

\newpage

\appendix
\section{State-of-the-Art for Distributed Dominating Set on Graphs With Arboricity at Most $\alpha$}\label{appendix:table}

\input{table-comprehensive}

\end{document}

%% file: commands.tex
\usepackage[utf8]{inputenc}

\usepackage[ruled,vlined,linesnumbered,noend]{algorithm2e}
\SetKwInOut{Parameter}{Parameters}

\usepackage{amsthm}
\usepackage{amsmath}
\usepackage{amssymb}
\usepackage{enumerate}
\usepackage{graphicx}
\usepackage[margin=1in]{geometry}
\usepackage{dsfont}
\usepackage{booktabs} 
\usepackage{tabularx}
\usepackage{multirow}
\usepackage{enumitem}
\usepackage{tablefootnote}

\usepackage{soul}

\usepackage[dvipsnames]{xcolor}
\definecolor{ForestGreen}{rgb}{0.1333,0.5451,0.1333}
\definecolor{DarkRed}{rgb}{0.65,0,0}

\usepackage{hyperref}
\hypersetup{
    colorlinks=true,
    linkcolor=DarkRed,
    filecolor=magenta,      
    urlcolor=cyan,
    citecolor=ForestGreen,
}

\usepackage{cleveref}

\usepackage{framed}
\usepackage[framemethod=tikz]{mdframed}
\usepackage{tikz-cd}
\usetikzlibrary{positioning,calc,arrows.meta}

\usepackage{makecell}

\newenvironment{wrapper}[1]
{
	\begin{center}
		\begin{minipage}{\linewidth}
			\begin{mdframed}[hidealllines=true, backgroundcolor=gray!20, leftmargin=0cm,innerleftmargin=0.4cm,innerrightmargin=0.4cm,innertopmargin=0.4cm,innerbottommargin=0.4cm,roundcorner=10pt]
				#1}
			{\end{mdframed}
		\end{minipage}
	\end{center}
}

\DeclareMathAlphabet{\mathmybb}{U}{bbold}{m}{n}

\newtheorem{theorem}{Theorem}[section]
\newtheorem{lemma}[theorem]{Lemma}

\newtheorem{corollary}[theorem]{Corollary}

\newtheorem{remark}[theorem]{Remark}

\newtheorem{claim}[theorem]{Claim}

\newcommand{\polylog}[0]{\mathrm{polylog}}

\newcommand{\congest}[0]{$\mathsf{CONGEST}$\xspace}
\newcommand{\local}[0]{$\mathsf{LOCAL}$\xspace}

%% file: introduction.tex
\section{Introduction}

The minimum dominating set (MDS) problem is one of the classic and well-studied problems in graph algorithms, where given a graph $G = (V,E)$, the goal is to find a set $D \subseteq V$ of the smallest possible size so that every node is dominated by $D$, i.e., either it is in $D$ or has a neighbor in $D$.
In the centralized setting, the well-known greedy algorithm obtains $\log(\Delta+1)$ approximation for graphs of maximum degree $\Delta$ \cite{Joh74}, which is optimal up to constant factors, as it is NP-Hard to achieve $c \log \Delta$ approximation for some constant $c$.

\paragraph{Distributed Model.}
In the distributed setting, there is a communication network with $n$ nodes that is identical to the input graph.
Nodes communicate with each other in synchronous rounds, where the message size can be bounded by $O(\log n)$ bits (known as \congest model) or unbounded (known as \local model).
Every node knows its neighbors in the beginning.

\paragraph{Distributed MDS on General Graphs.}
In the distributed setting, the first efficient algorithm for MDS was introduced by Jia, Rajaraman, and Suel \cite{JRS02} who gave a randomized $O(\log \Delta)$ expected approximation in $O(\log n \log (\Delta+1))$ rounds w.h.p.~in the \congest model.
Kuhn and Wattenhofer \cite{KW03} gave a $O(k\Delta^{2/k}\log \Delta)$ expected approximation in $O(k^2)$ rounds for any arbitrary $k$ in the \congest model. 
This is the first non-trivial approximation in a constant number of rounds when $k = O(1)$.
Later, Kuhn, Moscibroda, and Wattenhofer \cite{KMW06} gave a randomized $(1+\varepsilon)(1 + \ln(\Delta+1))$ approximation in $O(\varepsilon^{-4}\log^2\Delta)$ and $O(\varepsilon^{-2}\log n)$ rounds in the \congest and the \local model respectively.

In the deterministic world, Deurer, Kuhn, and Maus \cite{DKM19} provided two deterministic $(1+\varepsilon)(1 + \ln(\Delta+1))$ approximation algorithms for any $\varepsilon > 1/\polylog\Delta $ in the \congest model with $2^{O(\sqrt{\log n \log\log n})}$ and $O((\Delta + \log^* n) \cdot \polylog\Delta)$ round complexity respectively.
The round complexity $2^{O(\sqrt{\log n \log\log n})}$ in \cite{DKM19} can be reduced to $O(\mathrm{poly} \log n)$ by using improved deterministic algorithms such as \cite{RG20} for computing a network decomposition of $G^2$.\footnote{$G^2 = (V, E')$, where $V$ is the same vertex set of $G$, and $(u,v) \in E'$ if and only if the distance of $u$ and $v$ in $G$ is at most $2$.} 
Faour, Ghaffari, Grunau, Kuhn, and Rozho\v{n} \cite{FGGKR23} achieved a deterministic $O(\log \Delta)$ approximation in $\tilde{O}(\log^3 \Delta + \log^* n)$ rounds.\footnote{Their algorithm is primarily designed for the minimum set-cover problem, which concludes the mentioned result for MDS.} 
If we allow exponential-time local computation, Ghaffari, Kuhn, and Maus \cite{GKM17} showed that it is possible to achieve $(1+o(1))$ approximation in poly-logarithmic rounds in the \local model, which can be derandomized with the network decomposition of \cite{RG20}.
On the lower bound side, Kuhn, Moscibroda, and Wattenhofer \cite{KMW16} showed that achieving a poly-logarithmic approximation in the \local model requires 
$$\min \{\Omega( \log \Delta/ \log\log \Delta ), \Omega( \sqrt{\log n/ \log \log n})\}$$ rounds.

\paragraph{Distributed MDS on Special Graph Families.}
The MDS problem has been under significant attention on a variety of special graph families, such as planar graphs, bounded genus graphs, bounded expansion graphs, bounded arboricity graphs (which is the focus of our paper), and more (see, e.g.~\cite{Bak94, LOW08, CHW08, BE08, LW10, LPW13, Waw14, AORS18, ASS19, AFK20, AJ20, Amiri21, BCGW21, GMP22, DGI22, HKOSV25}).

\paragraph{Graphs of Bounded Arboricity.}
The arboricity of a graph is the smallest number of forests into which the edges of the graph can be partitioned.
It is believed that some natural real world graphs, such as the world wide web graph, social networks, and transaction networks have bounded arboricity.
This led to extensive research on this family of graphs not only in the distributed setting but also in other settings such as streaming, dynamic, and sublinear 
(see e.g.~\cite{Epp94, BF99,GG06,KMW06,LW10,HTZ14,PS16,PPS16,BU17,OSSW18,ELR18,MV18,SW19,ERR19,ERS20,BPS20,ELR20,BS20,BCG20,KS21, MSW21,GBSH21,CCMU21,DGI22,ERR22,BO22,ELR24,sKMST24,GG24,BCPS24,BCPSESA24,Gil24,BCM25}).

\paragraph{Distributed MDS on Graphs of Bounded Arboricity.}
The first algorithm for MDS on bounded arboricity graphs was introduced by Lenzen and Wattenhofer~\cite{LW10}, who gave two algorithms, a randomized $O(\alpha^2)$ approximation and a deterministic $O(\alpha  \log \Delta)$ approximation in $O(\log n)$ and $O(\log \Delta)$ round complexity respectively. 
Later, Amiri \cite{Amiri21} provided a deterministic $O(\alpha^2)$ approximation in $O(\log n)$ rounds.
None of these algorithms require prior knowledge of the value of $\alpha$.
Bansal and Umboh \cite{BU17} showed that it is NP-Hard to achieve $\alpha-1-\varepsilon$ approximation and provided a centralized $3\alpha$ approximation LP based algorithm that can be improved to $(2\alpha+1)$ approximation by optimizing the parameters \cite{Dvo19}.
Together with the $(1+\varepsilon)$ approximation of Kuhn, Moscibroda, and Wattenhofer \cite{KMW06} for solving the LP, it is possible to implement the centralized algorithm of \cite{BU17} in the \congest model to achieve $(2\alpha+1)(1+\varepsilon)$ approximation in $O(\varepsilon^{-4}\log^2\Delta)$ rounds.
This algorithm needs prior knowledge of the value of $\alpha$.
Morgan, Solomon, and Wein \cite{MSW21} provided an $O(\alpha)$ approximation in $O(\alpha  \log \Delta)$ rounds that needs prior knowledge of $\alpha$.
Dory, Ghaffari, and Ilchi \cite{DGI22} provided a deterministic $(2\alpha+1)(1+\varepsilon)$ approximation and a randomized $\alpha(1+o(1))$ approximation in $O(\varepsilon^{-1}\log\Delta)$ and $O(\alpha  \log \Delta/\log \alpha)$ rounds respectively.
These algorithms \textbf{need prior knowledge} of $\alpha$.
\cite{DGI22} also described a deterministic $(4\alpha+2)(1+\varepsilon)$ approximation in $O(\varepsilon^{-1}\log n)$ rounds without prior knowledge of $\alpha$.
All of the aforementioned results work in the \congest model.

On the lower bound side, \cite{DGI22} showed that one needs at least $\Omega(\log\Delta/\log\log\Delta)$ rounds to obtain any constant or poly-logarithmic approximation, even in the \local model and even in the case where $\alpha=2$.
Although many of the previous results have a near-optimal round complexity of $O(\log\Delta)$, one can ask the following natural question.

\begin{wrapper}
\begin{center}
\textbf{Question.} What is the best approximation ratio achievable for MDS in the optimal round complexity of $O(\log\Delta/\log\log\Delta)$?
\end{center}
\end{wrapper}

The same question has been studied for related classic problems, such as \textbf{minimum vertex-cover} (more generally, \textbf{minimum set-cover}), 
and \textbf{maximum matching}
(See \cite{BYCHGS17,EGM18,BBEKS19}).

\subsection{Technical Overview}

Our main result is as follows.

\begin{theorem}\label{thm:mainn}
    There is a deterministic $O\!\left(\alpha \cdot \frac{\log \Delta}{\log\log\Delta}\right)$-~approximation algorithm for the minimum dominating set problem in graphs with arboricity at most $\alpha$ and maximum degree $\Delta$ that has optimal round complexity and message size, i.e., the algorithm runs in $O\!\left(\frac{\log \Delta}{\log\log\Delta}\right)$ rounds, and the message size is $1$ bit per round.
    Moreover, the algorithm does not require prior knowledge of $\alpha$, but it requires prior knowledge of $\Delta$.
\end{theorem}

Our algorithm runs in the optimal number of rounds among all algorithms that achieve a poly-logarithmic approximation factor due to the lower bound of \cite{DGI22}.
The message size of our algorithm is optimal as well.
Moreover, our algorithm has the best approximation ratio among all previous results in the optimal round complexity of $O(\log\Delta/\log\log\Delta)$, when the value of $\alpha$ is unknown, but the value of $\Delta$ is known to the algorithm.
We achieve our result by providing a deterministic algorithm for the distributed minimum dominating set problem on graphs of bounded arboricity that has a trade-off between the approximation ratio and the round complexity in \Cref{thm:main1}.

\begin{theorem}\label{thm:main1}
For any $\beta > 1$, there is a deterministic $10\alpha(\beta + \log_\beta (\Delta+1))$-approximation algorithm for the minimum dominating set problem in graphs with arboricity at most $\alpha$ and maximum degree $\Delta$ that runs in $4\!\left\lceil\log_\beta (\Delta+1)\right\rceil$ rounds.
Moreover, the message size of the algorithm is only one bit per round, and the algorithm does not require prior knowledge of $\alpha$, but it requires prior knowledge of $\Delta$.
\end{theorem}

By plugging $\beta = \Theta \left( \frac{\log\Delta}{\log\log\Delta} \right)$ in the above theorem, we get $O\!\left(\alpha \cdot \frac{\log \Delta}{\log\log\Delta}\right)$-approximation in $O\!\left(\frac{\log \Delta}{\log\log\Delta}\right)$ rounds, which concludes \Cref{thm:mainn}.
Our algorithm for \Cref{thm:main1} works as follows.
It runs in $\log_\beta(\Delta+1)$ iterations from $i = 1$ up to $\log_\beta(\Delta+1)$, and builds a dominating set $S$.
At iteration $i$, every undominated (by $S$) node $u \in V$ votes for one of its neighbors that dominates at least $(\Delta+1)/\beta^i$ undominated nodes (if any).
Then, every node that has been voted at least once, becomes part of the solution $S$.

\medskip
\noindent
\textbf{Comparison to \cite{LW10}.}
One of the algorithms by Lenzen and Wattenhofer  \cite{LW10} achieves a trade-off between the approximation ratio and the round complexity as follows.
For any $\beta \geq 2$, it is possible to get $O(\alpha\beta\log_\beta \Delta)$ approximation in $O(\log_\beta \Delta)$ rounds with message size bounded by $O(\log\log_\beta\Delta)$.\footnote{This message size is not explicitly mentioned in Corollary 3 of their paper (the mentioned algorithm), but it is easy to verify, since every node $u$ will send an integer bounded by $O(\log_\beta \Delta)$.
One needs only $O(\log \log_\beta \Delta)$ bits to represent this integer.}
Hence, by letting $\beta = \log^\epsilon \Delta$, one achieves a $O(\alpha \varepsilon^{-1} \log^{1+\varepsilon}\Delta/\log\log\Delta)$ approximation in $O(\varepsilon^{-1}\log\Delta/\log\log\Delta)$ rounds with $O(\log(\varepsilon^{-1}\log\Delta))$ message size.

The algorithm works as follows.
Every undominated node $u\in V$ considers its residual degree (defined as the number of undominated nodes that can be dominated by $u$ at this point in time), and if its residual degree is within a factor $2$ of the maximum residual degree of all nodes in its 2-hop neighborhood (nodes of distance at most $2$ of $u$), then $u$ becomes one of the candidates to be part of the solution.
Then, each undominated node votes for one of the candidates in its neighborhood (if any) arbitrarily, and the voted node becomes part of the solution.

Our algorithm is a simplification of this result.
Instead of considering the 2-hop neighborhood of a node $u$ to determine if $u$ is a candidate, we simply consider each $u$ as a candidate if it dominates $\tau$ many undominated nodes at this point in time, where $\tau$ is a threshold that decreases geometrically over time.
We provide the approximation ratio analysis of this simplified algorithm, which is also stronger than the approximation ratio of the specified algorithm of \cite{LW10}.
More precisely, the approximation ratio is reduced from $O(\alpha\beta\log\Delta)$ to $O(\alpha(\beta+\log\Delta))$.

The only downside of our algorithm compared to \cite{LW10} is the assumption of prior knowledge of $\Delta$.
Let us also mention that, in Theorem 2 of \cite{LW10}, it is stated that the approximation ratio of their algorithm is tight.
As a result, some modifications to \cite{LW10} is necessary to hope for a better approximation ratio.

\medskip
\noindent
\textbf{Comparison to \cite{Amiri21} and \cite{DGI22}.}
In the case of unknown $\alpha$, both of \cite{Amiri21} and \cite{DGI22} first apply the algorithm of Barenboim and Elkin \cite{BE08} to compute a so called \textit{forest decomposition}, i.e.~an orientation of the edges of the graph is computed in $O(\varepsilon^{-1}\log n)$ rounds such that the out-degree of each node is at most $(2+\varepsilon)\alpha$.
It is not obvious at all how we can get around this step since both of the algorithms in \cite{Amiri21} and \cite{DGI22} heavily depend on this step.
In \cite{Amiri21}, the algorithm needs this decomposition to construct a set-cover instance whose solution translates to a minimum dominating set in the original graph.
In \cite{DGI22}, every node computes a local approximation of the arboricity according to the decomposition and uses this approximate value for its own decisions (See Remark 4.5 in the arXiv version of \cite{DGI22} for more details).
On the other hand, our algorithm does not need the decomposition or any information about $\alpha$.
We should mention that the algorithm of \cite{DGI22} is simple and elegant, and if we assume \textbf{prior knowledge} of $\alpha$, one can get a $O(\alpha\log^{\varepsilon}\Delta)$ approximation in $O(\varepsilon^{-1}\log \Delta / \log\log \Delta)$ rounds.
In \cite{DGI22} the result is stated as $(2\alpha+1)(1+\varepsilon)$ approximation in $O(\varepsilon^{-1}\log \Delta)$ rounds, but the actual round complexity is $\log_{1+\varepsilon} \Delta$ and the parameter $\varepsilon$ here does \textbf{not} need to be close to zero. It remains to replace $1 + \varepsilon \gets \log^\varepsilon \Delta$ to get the mentioned result.

\paragraph{Summary of Known Results.}
\Cref{tab:MDS-without} summarizes the results for MDS on graphs of arboricity $\alpha$ in the case where the value of $\alpha$ is \textbf{unknown} to the algorithm.
A more comprehensive table is provided in \Cref{appendix:table}.

\input{table-state-of-the-art}

\subsection{Preliminaries}\label{sec:pre}

Let $G = (V, E)$ be a simple undirected graph.
For every node $v \in V$, we denote by $N(v)$ the set of neighbors of $v$, and define $N^+(v) := N(v) \cup \{v\}$.
For each subset of nodes $U \subseteq V$, we define $N^+(U) = \cup_{u \in U} N^+(u)$.
For each two subsets of nodes $U,C \subseteq V$, we define $N^+(U,C) := N^+(U) \cap C$.
For every $U \subseteq V$, let $E[U]$ be the set of edges whose both endpoints belong to $U$, i.e.~$E[U] := \{uv \in E \mid u,v \in U\}$.
We denote the arboricity of $G$ by $\alpha(G)$, which is defined as the minimum number of forests into which $E$ can be partitioned.
According to the well-known Nash-Williams' theorem \cite{NashWilliams}, we have
\begin{equation}\label{eq:arboricity-formula}
   \alpha(G) = \max_{U \subseteq V, \ |U|\geq 2} \left\lceil \frac{|E[U]|}{|U|-1} \right\rceil. 
\end{equation}

%% file: table-state-of-the-art.tex
\renewcommand{\arraystretch}{1.9}
\begin{table}[h!]
\caption{State-of-the-art for distributed minimum dominating set on graphs with arboricity at most $\alpha$, where the value of $\alpha$ is \textbf{unknown} to the algorithm. $\beta \geq 2$ and $1 > \varepsilon > 0$ are arbitrary parameters. \cite{DGI22} works for the weighted case.}\label{tab:MDS-without}
\vspace{0.3cm}
\centering
\begin{tabular}{ |c|c|c|c| }
 \hline
 Approximation & Rounds & Deterministic & Reference \\
 \hline
 $O(\alpha^2)$ & $O(\log n)$ & $\times$  & \multirow{2}{3.1em}{\cite{LW10}} \\ 
 $O(\alpha \beta \log_\beta \Delta)$ & $O(\log_\beta \Delta)$ & \checkmark & \\
 \hline
 $O(\alpha^2)$ & $O(\log n)$ & \checkmark & \cite{Amiri21} \\
 \hline
 $(4\alpha+2)  (1+\varepsilon)$ & $O\left(\frac{\log n}{\varepsilon}\right)$ & \checkmark & \cite{DGI22} \\
 \hline
 $O(\alpha(\beta + \log_\beta \Delta))$ & $O(\log_\beta \Delta)$ & \checkmark & \Cref{thm:main1} \\
 \hline
\end{tabular}
\end{table}
\renewcommand{\arraystretch}{1}

%% file: algorithm.tex
\section{Achieving $O\!\left(\alpha (\beta + \log_\beta \Delta)\right)$-Approximation in $O\!\left( \log_\beta \Delta\right)$ Rounds}

In this section, we provide the algorithm for \Cref{thm:main1} together with its analysis.

\subsection{Description of the Algorithm}

Initially, we let $S = \emptyset$ and all nodes are undominated.
Throughout the algorithm, $S$ represents the set of selected nodes so far to be part of the output dominating set, and a node $u$ is dominated if and only if $u \in N^+(S)$.
Fix any arbitrary parameter $\beta > 1$.
The algorithm runs in $ \ell := \lceil \log_\beta (\Delta+1) \rceil $ iterations.
In the $i^\mathrm{th}$ iteration, every undominated node $u \notin N^+(S)$ votes for one node $p_u$ in its neighborhood $p_u \in N^+(u)$ that covers at least $(\Delta+1) / \beta^i$ undominated nodes (if any).
Then $p_u$ will be added to the solution $S$.

More precisely, at the beginning of iteration $i$, every node $u$ knows which of its neighbors are dominated.
Then, we have the following procedure during iteration $i$.
\begin{enumerate}
    \item Every node $u$ computes $\deg^{(i)}(u) := |N^+(u) \setminus N^+(S)|$, i.e.~the number of undominated nodes in $N^+(u)$.
    Then, $u$ informs its neighbors, whether or not this value is greater than or equal to $(\Delta+1)/\beta^i$.
    
    \item Every \textbf{undominated} node $u$ considers all elements in $N^+(u)$ and votes for a node $p_u \in N^+(u)$ that has $\deg^{(i)}$ value greater than or equal to $(\Delta+1)/\beta^i$, breaking the ties arbitrarily.
    Note that $p_u$ might be equal to $u$ itself.
    If no such element exists (which means for all $v \in N^+(u)$ we have $\deg^{(i)}(v) < (\Delta+1)/\beta^i$), we let $p_u = \bot$.
    Now, if $p_u \neq \bot$, then $u$ informs $p_u$ that it is voted to be part of the solution.

    \item Each voted node $v$ (i.e., there exists an undominated $u$ such that $v=p_u$) informs all of its neighbors that $v$ is voted and is now part of the solution $S$.
    Hence, $S \gets S + v$, and the neighbors of $v$ mark themselves as dominated.
    Note that this step happens simultaneously, and many nodes might be added to $S$.

    \item Every node $u$ informs its neighbors whether $u$ is dominated or undominated at this point.
\end{enumerate}

\begin{remark}
    Before we analyze our algorithm, let us mention that the voting system in our algorithm is necessary.
    More precisely, if there is no voting and each vertex $u \in V$ satisfying $\deg^{(i)}(u) \geq (\Delta + 1)/\beta^i$ directly becomes part of the solution, then we can not bound the approximation ratio.
    The issue arises in the last iterations of the algorithm, roughly in iterations $i \geq i^\star$ where $(\Delta+1)/\beta^{i^\star} \approx \alpha$.
    At those iterations, it can happen that the set of nodes $M$ dominating $(\Delta+1)/\beta^{i}$ many undominated nodes has a much larger size than the number of undominated nodes themselves.
    Hence, although we can just insert all undominated nodes into the solution, adding $M$ to the solution is far from optimal.
    Since the value of $\alpha$ is unknown to the algorithm, we perform the voting in all iterations to avoid the issue.
\end{remark}

\subsection{Analysis}

\paragraph{Correctness.}
According to Step 4 of the iteration, it is obvious that at the beginning of each iteration, every node $u$ correctly knows which of its neighbors are dominated.
It is easy to see that, conditioned on the previous knowledge, the values of $\deg^{(i)}$ are computed correctly, every node correctly knows whether or not it is a candidate, every node correctly knows whether it has been voted, and every node knows whether it is dominated correctly at the end of the iteration.

The algorithm runs while $i \leq \ell = \lceil \log_\beta(\Delta+1) \rceil$.
Note that we have $(\Delta+1)/ \beta^{\ell} \leq 1 $.
This concludes that in the last iteration of the algorithm, for every undominated node $u$, we have $p_u \neq \bot$ (since at least $u$ itself satisfies $\deg^{(i)}(u) \geq 1 \geq (\Delta+1)/ \beta^{\ell}$), and $u$ will be dominated at the end of this iteration.
This concludes that the final solution $S$ of the algorithm is indeed a dominating set.

\paragraph{Round Complexity.}
Clearly, each of the Steps in every iteration can be done in only one round simultaneously for all nodes.
This concludes that the algorithm runs in $4 \cdot \left\lceil \log_\beta (\Delta+1) \right\rceil$ rounds.

\paragraph{Message Sizes.}
The size of the messages between each pair of neighbors in each round is only one bit, since in each step of each iteration, every node $u$ only needs to notify its neighbors about the correctness of a condition.
In the first step, every node $u$ only needs to notify its neighbors whether or not the condition $\deg^{(i)}(u) \geq (\Delta+1)/\beta^i$ holds.
In the second step, every node $u$ only needs to notify $p_u$ that it has been voted for by $u$ (and notify other neighbors that they have not been voted for by $u$).
In the third step, each node $u$ notifies its neighbors whether or not $u \in S$.
In the fourth step, each node $u$ notifies its neighbors whether or not $u$ is dominated.
Hence, every round of the algorithm only needs one bit of communication for each pair of neighbors.

\subsection{Approximation Ratio Analysis}

Let $D^\star$ be an optimum dominating set.
We consider a partition on the solution $S = S_1 \cup \cdots \cup S_\ell$ (recall that $\ell = \lceil \log_\beta(\Delta+1) \rceil$), where $S_i$ consists of the nodes that became part of the solution during iteration $i$ of the algorithm.
For simplicity, define $S_0 := \emptyset$.
For each $1 \leq i \leq \ell$, let $C_i$ be the set of nodes that became dominated at the end of iteration $i$ of the algorithm, i.e.,~$ C_i := N^+(S_i) \setminus N^+(S_0 \cup \cdots \cup S_{i-1}) $.
Hence, $C_1, \ldots, C_\ell$ form a partition of $V$.
For simplicity, define $C_0 := \emptyset$.

The analysis consists of two parts. In the first part, we show that at the end of each iteration, the number of newly dominated nodes is proportional to the number of nodes that become part of the solution, i.e.~$|S_i|$ is proportional to $|C_i|$.
This shows that the choices of the algorithm were good during iteration $i$.
In the second part, we analyze the contribution of the optimum solution while dominating each $C_i$.
We start with the following claim.

\begin{claim}\label{claim:S-C}
    If $S, C \subseteq V$ and  each $s \in S$ has at least $k$ neighbors in $C$, then we have
    $$ \left( \frac{k}{2\alpha}-1 \right) \cdot |S| \leq |C|. $$
\end{claim}

\begin{proof}
    There are at least $(k/2) \cdot |S|$ edges between nodes in $S \cup C$ since each node in $S$ has at least $k$ neighbors in $C$, and each edge is counted at most twice (if both of its endpoints belong to $S \cap C$).
    Hence,
    \begin{equation*}
        \frac{k \cdot |S|}{2} \leq |E[S \cup C]| \leq \alpha \cdot (|S \cup C| - 1) \leq \alpha \cdot (|S| + |C|),
    \end{equation*}
    where the second inequality follows by \Cref{eq:arboricity-formula}.
    This concludes the claim.
\end{proof}

\begin{corollary}\label{cor:Si-compare-to-Ci}
    For each $1 \leq i \leq \ell$, we have
    $$ \frac{\Delta+1}{5\alpha \beta^{i}} \cdot |S_i| \leq |C_i|. $$
\end{corollary}

\begin{proof}
We consider two cases.
\begin{itemize}
    \item $i$ is small, more precisely, $(\Delta+1) / (5\alpha\beta^{i}) \geq 1 $.
    According to the algorithm, each $s \in S_i$ has at least $(\Delta+1) / \beta^{i} - 1$ neighbors in $C_i$ (the undominated elements of $N^+(s)$, excluding $s$ itself at the beginning of iteration $i$).
    Thus, \Cref{claim:S-C} and the assumption of this case conclude
    \begin{equation*}
        \frac{\Delta+1}{5\alpha\beta^{i}}  \cdot |S_i|
        \leq \left( \frac{(\Delta+1)/\beta^i - 1}{2\alpha}-1 \right) \cdot |S_i| \leq |C_i|.
    \end{equation*}
    
    \item $i$ is large, more precisely $(\Delta+1) / (5\alpha\beta^{i}) \leq 1 $.
    According to the algorithm, we have $|S_i| \leq |C_i|$ since for each $v \in S_i$, there exists at least one $u \in C_i$ (where $v=p_u$) that voted for $v$, and each $u$ votes for at most one node $p_u \in S_i$.
    Combining with the assumption of this case, we have
    $$ \frac{\Delta+1}{5\alpha\beta^{i}}  \cdot |S_i| \leq |S_i| \leq |C_i|. \qedhere$$
\end{itemize}
\end{proof}

For each $1 \leq i \leq \ell$, assume that $D_i$ is the set of nodes in the optimum solution $D^\star$ that dominate at least one node in $C_i$, but do not dominate any node in $C_0 \cup \cdots \cup C_{i-1}$, i.e.,~$D_i := D^\star \cap \left( N^+(C_i) \setminus N^+(C_0 \cup \cdots \cup C_{i-1}) \right)$.
It is easy to see that $D_1, \ldots, D_\ell$ is a partition of $D^\star$.
For each $1 \leq i \leq \ell$, we split $D_i$ into two subsets $D_i'$ and $D_i'':= D_i - D_i'$ as follows.
$D_i'$ is the set of nodes in $D_i$ that does not dominate any node outside $C_{i}$, i.e., $D_i' := D_i \cap \left(  N^+(C_i) \setminus  N^+(C_{i+1} \cup \cdots \cup C_{\ell}) \right)$.
Intuitively, $D_i'$ does not contribute to dominating any node outside $C_i$.
But, $D_i''$ has the potential to contribute to dominating nodes in $C_j$ for $j > i$.
\Cref{fig} illustrates the connection between these objects.
Note that for the last index $i = \ell$, we always have $D_\ell'' = \emptyset$ since $D^\star$ is a dominating set.

\input{fig}

\begin{lemma}\label{lem:Ci-compare-to-Di}
    For each $1 \leq j \leq \ell$, we have
    $$ |C_j| \leq \frac{\Delta+1}{\beta^{j-1}} \cdot |D_j'| + \sum_{i=1}^j |N^+(D_i'', C_j)|. $$
\end{lemma}

\begin{proof}
    We claim that the number of nodes in $C_j$ that can be dominated by an arbitrary $v \in D_j'$ is at most $\frac{\Delta+1}{\beta^{j-1}}$. 
    The case $j=1$ follows from $\deg(v) \leq \Delta$.
    For $j > 1$, assume $ |C_j \cap N^+(v)| >\frac{\Delta+1}{\beta^{j-1}} $.
    This means that in the $(j-1)^{\text{th}}$ iteration of the algorithm, we have $\deg^{(j-1)}(v) > \frac{\Delta+1}{\beta^{j-1}}$.
    Hence, $v$ voted for $p_v$ and $v$ became dominated in iteration $(j-1)$ of the algorithm, which is in contradiction with the assumption that $v \in C_j$.
    As a result, each $v \in D_j'$ can dominate at most $\frac{\Delta+1}{\beta^{j-1}}$ elements of $C_j$, which concludes 
    $$|N^+(D_j', C_j)| \leq \frac{\Delta+1}{\beta^{j-1}} \cdot |D_j'| . $$
    The number of nodes in $C_j$ that can be dominated by $\cup_{i=1}^j D_i''$ is trivially at most
    $\sum_{i=1}^j |N^+(D_i'', C_j)|$.
    According to the definition of the partitioning on $D^\star$, each node $v \in D^\star$ satisfying $N^+(v) \cap C_j \neq \emptyset$ (contributing to dominating $C_j$) belongs to at least one of $D_j'$ and $\cup_{i=1}^j D_i''$.
    Since $D^\star$ is a dominating set, it should dominate the whole $C_j$, which concludes
    \begin{align*}
        |C_j| &= |N^+(D^\star, C_j)| 
        \leq |N^+(D_j', C_j)| + \sum_{i=1}^j |N^+(D_i'', C_j)| \\
        &\leq \frac{\Delta+1}{\beta^{j-1}}  \cdot |D_j'| + \sum_{i=1}^j |N^+(D_i'', C_j)|. \qedhere
    \end{align*}
\end{proof}

\begin{lemma}\label{lem:upper-bound-for-S}
We have
    $$ |S| \leq 5\alpha\beta \cdot |D'| +  \frac{5\alpha}{\Delta+1}\cdot \sum_{j=1}^\ell \sum_{i=1}^j \beta^j \cdot |N^+(D_i'', C_j)|, $$
    where $D' = \cup_{j=1}^\ell D_j'$.
\end{lemma}

\begin{proof}
    According to \Cref{cor:Si-compare-to-Ci} and \Cref{lem:Ci-compare-to-Di}, for each $1 \leq j \leq \ell$, we have
    \begin{align*}
    |S_j|  
    \leq \frac{5\alpha \beta^j}{\Delta+1} \cdot |C_j| &\leq  \frac{5\alpha \beta^j}{\Delta+1} \cdot \left(  \frac{\Delta+1}{\beta^{j-1}} \cdot |D_j'| + \sum_{i=1}^j |N^+(D_i'', C_j)| \right)  \\
    &\leq 5\alpha\beta \cdot |D_j'| + \frac{5\alpha}{\Delta+1} \cdot \sum_{i=1}^j \beta^j \cdot |N^+(D_i'', C_j)| .
    \end{align*}
    The lemma follows by summing up these inequalities for all $1 \leq j \leq \ell$.
\end{proof}

\begin{lemma}\label{lem:v-star}
    For each $1 \leq i \leq \ell$ and every $v^\star \in D_i''$, we have
    $$  \sum_{j=i}^\ell \beta^{j} \cdot |N^+(v^\star, C_j)| \leq (\log_\beta(\Delta+1) + \beta) \cdot (\Delta+1). $$
\end{lemma}

\begin{proof}
    Since $D''_\ell = \emptyset$, we can simply assume that $i < \ell$.
    Let $x_j := |N^+(v^\star, C_j)|$ for each $i \leq j \leq \ell$.
    Assume $t$ is the maximum index such that $N^+(v^\star) \cap C_t \neq \emptyset$.
    Note that according to the definition of $D_i''$, we have $t \geq i + 1$.
    Now, consider an arbitrary $i \leq k \leq t-1$.
    We claim that
    \begin{equation}\label{eq:sum-xi}
     x_k + x_{k+1} + \cdots + x_t \leq (\Delta+1)/\beta^{k}    .
    \end{equation}
    For the sake of contradiction, assume this is not the case.
    Since none of the nodes in $C_k \cup C_{k+1} \cup \cdots \cup C_t$ are dominated at the beginning of iteration $k$ of the algorithm, we conclude 
    $$\deg^{(k)}(v^\star) \geq x_k + x_{k+1} + \cdots + x_t > (\Delta+1)/\beta^{k} . $$
    Now, consider an arbitrary $ u \in N^+(v^\star) \cap C_{t}$. This implies that $v^\star \in N^+(u)$ has $\deg^{(k)}(v^\star) > (\Delta+1)/\beta^{k}$, which means $u$ voted for $p_u$ (which also has $\deg^{(k)}$ at least $(\Delta+1)/\beta^{k}$) in iteration $k$ of the algorithm and this is in contradiction with assuming $u \in C_{t}$ and $t \geq k+1$ as $u$ would become dominated in iteration $k$ not $t$. As a result, \Cref{eq:sum-xi} holds.

   We conclude that for each $i \leq k \leq t-1$, we have $x_k \leq (\Delta+1)/\beta^k$ and for the special case $k = t-1$, we also conclude that $x_t \leq x_{t-1} + x_t \leq (\Delta+1)/\beta^{t-1}$
   (Note that the inequality $x_t \leq (\Delta+1)/\beta^t$ is \textbf{not} necessarily true, and we lose a factor $\beta$ to bound the last $x_t$).
   Finally,
    \begin{align*}
        \sum_{j=i}^\ell \beta^j \cdot x_j 
        &= 
        \left(\sum_{j=i}^{t-1} \beta^j \cdot x_j \right)+ \beta^t \cdot x_t
        \leq
        \left( \sum_{j=i}^{t-1} (\Delta+1) \right) + \beta \cdot (\Delta+1) \\
        &\leq 
        (\ell -1 + \beta )\cdot (\Delta+1) \leq (\log_\beta(\Delta+1) + \beta) \cdot (\Delta+1),
    \end{align*}
    where we used $x_{t+1} = \cdots = x_\ell = 0$ according to the maximality of $t$.
\end{proof}

\begin{corollary}\label{cor:sum-neighbors-bound}
    We have
    $$  \sum_{i=1}^\ell \sum_{j=i}^\ell \beta^{j} \cdot |N^+(D_i'', C_j)| \leq \left( \log_\beta(\Delta+1) + \beta \right) \cdot (\Delta+1) \cdot |D''|, $$
    where $D'' = \cup_{i=1}^\ell D_i''$.
\end{corollary}

\begin{proof}
    Fix an $1 \leq i \leq \ell$. Since $N^+(D_i'', C_j) = \cup_{v \in D_i''} N^+(v,C_j)$, by \Cref{lem:v-star} for all $v \in D_i''$, we conclude that
    \begin{align*}
    \sum_{j=i}^\ell \beta^{j} \cdot |N^+(D_i'', C_j)| 
    &\leq \sum_{j=i}^\ell \beta^{j} \cdot \sum_{v \in D_i''} |N^+(v, C_j)| \\
    &= \sum_{v \in D_i''} \sum_{j=i}^\ell \beta^{j} \cdot |N^+(v, C_j)| \\
    &
    \leq \sum_{v \in D_i''} (\log_\beta(\Delta+1) + \beta ) \cdot (\Delta+1) \\
    &= (\log_\beta(\Delta+1) + \beta ) \cdot (\Delta+1) \cdot |D_i''| .
    \end{align*}
    Summing up this inequality for all $1 \leq i \leq \ell$ concludes the corollary.
\end{proof}

Finally, we get an upper bound on the approximation ratio of our algorithm as follows.
Combining \Cref{lem:upper-bound-for-S} and \Cref{cor:sum-neighbors-bound}, we have
\begin{align*}
    |S| 
    &\leq 
    5\alpha \beta \cdot |D'| + \frac{5\alpha}{\Delta+1}\cdot \sum_{j=1}^\ell \sum_{i=1}^j \beta^j \cdot |N^+(D_i'', C_j)| \\
    &= 5\alpha \beta \cdot |D'| + \frac{5\alpha}{\Delta+1}\cdot \sum_{i=1}^\ell \sum_{j=i}^\ell \beta^j \cdot |N^+(D_i'', C_j)| \\
    &\leq
    5\alpha\beta \cdot |D'| + \frac{5\alpha}{\Delta+1}\cdot 
    \left( \log_\beta(\Delta+1) + \beta \right) \cdot (\Delta+1) \cdot |D''| \\
    &\leq 10\alpha \cdot (\beta + \log_\beta (\Delta+1)) \cdot |D^\star|.
\end{align*}
This completes the approximation ratio analysis.

%% file: fig.tex
\begin{figure}[h!]
  \centering
\begin{tikzpicture}[node distance=1.5cm and 2cm,
  every node/.style={draw, minimum width=0.5cm, minimum height=0.4cm},
  -, >=Latex]

\node (D1p) {$D_1''$};
\node (C1) [below=of D1p, minimum width=1cm, minimum height=0.5cm] {$C_1$};
\node (D1pp) [node distance=0.0cm, below=of C1, yshift=-0.4cm] {$D_1'$};

\node (D2p) [right=of D1p] {$D_2''$};
\node (C2) [below=of D2p, minimum width=1cm, minimum height=0.5cm] {$C_2$};
\node (D2pp) [node distance=0.0cm, below=of C2, yshift=-0.4cm] {$D_2'$};

\node (D3p) [right=of D2p] {$D_3''$};
\node (C3) [below=of D3p, minimum width=1cm, minimum height=0.5cm] {$C_3$};
\node (D3pp) [node distance=0.0cm, below=of C3, yshift=-0.4cm] {$D_3'$};

\node (DLp) [right=of D3p, xshift=1cm] {$D_\ell''$};
\node (CL) [below=of DLp, minimum width=1cm, minimum height=0.5cm] {$C_\ell$};
\node (DLpp) [node distance=0.0cm, below=of CL, yshift=-0.4cm] {$D_\ell'$};

\draw (D1pp) -- (C1);
\draw[bend left=3] (D1p) -- (C2);
\draw[bend left=3] (D1p) to (C3);

\draw (D2pp) -- (C2);
\draw[bend left=3] (D2p) -- (C3);

\draw (D3pp) -- (C3);

\draw (DLpp) -- (CL);

\draw[bend left=3] (D1p) to (CL);
\draw[bend left=3] (D2p) -- (CL);
\draw[bend left=3] (D3p) -- (CL);

\draw[bend right=5] (D1p) to (C1);
\draw[bend right=5] (D2p) to (C2);
\draw[bend right=5] (D3p) to (C3);
\draw[bend right=5] (DLp) to (CL);

\node[draw=none, fill=none] at ($(D3p)!0.5!(DLp)$) {\Large$\cdots$};
\node[draw=none, fill=none] at ($(C3)!0.5!(CL)$) {\Large$\cdots$};

\end{tikzpicture}
\caption{Connection between different parts of the optimal solution $D^\star$ and the $C_i$'s.
For each $i \leq j$ there might be an edge between $D_i''$ and $C_j$ in the graph, and for each $i > j$, there is no edge between $D_i''$ and $C_j$.
For each $i \neq j$, there is no edge between $C_i$ and $D_j'$ in the original graph.
Note that $D_i'$ or $D_i''$ can have an intersection with some of the $C_j$'s as well, but for simplicity, we drew them separately.}
\label{fig}
\end{figure}
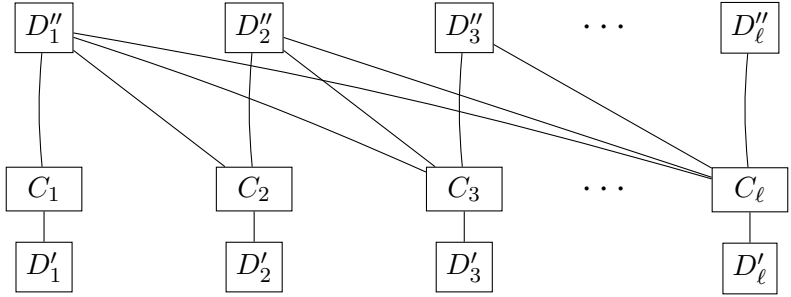

%% file: table-comprehensive.tex
\renewcommand{\arraystretch}{1.9}
\begin{table}[h!]
    \caption{State-of-the-art for distributed minimum dominating set on graphs with arboricity at most $\alpha$. Except for \cite{DGI22}, all of the upper bounds work only on the unweighted case.}
    \label{tab:MDS2}
    \vspace{0.3cm}
    \centering
    \begin{tabular}{ |c|c|c|c|c| }
     \hline
     Approximation & Rounds & Deterministic & Knowledge of $\alpha$ &  Paper \\
     \hline
     $O(\alpha^2)$ & $O(\log n)$ & $\times$  & $\times$ & \multirow{2}{3.2em}{\cite{LW10}} \\ 
     $O(\alpha \cdot \log \Delta)$ & $O(\log \Delta)$ & \checkmark & $\times$ & \\
\hline
     $O(\alpha^2)$ & $O(\log n)$ & \checkmark & \checkmark & \cite{Amiri21} \\     
\hline
     $(2\alpha+1)(1+\varepsilon)$ & $O(\varepsilon^{-4} \cdot \log^2 \Delta)$ & \checkmark & \checkmark & \cite{BU17, KMW06} \\
     \hline
     $O(\alpha)$ & $O(\alpha \cdot \log n)$ & $\times$ & $\times$ & \cite{MSW21} \\
     \hline
     $(2\alpha+1) \cdot (1+\varepsilon)$ & $O(\varepsilon^{-1} \cdot \log \frac{\Delta}{\alpha})$ & \checkmark & \checkmark & \multirow{4}{3.6em}{\cite{DGI22}} \\
     $(4\alpha+2) \cdot (1+\varepsilon)$ & $O(\varepsilon^{-1} \cdot \log n)$ & \checkmark & $\times$ &  \\
     Expected $\alpha+O(\log \alpha) $ & $O(\frac{\alpha \log \Delta}{\log \alpha})$ & $\times$ & \checkmark &  \\
     poly-logarithmic & $\Omega( \frac{\log \Delta}{\log\log \Delta})$ & - & $\alpha \geq 2$ &  \\
     \hline
     $\alpha - 1 - \varepsilon$ & NP-Hard & - & - & \cite{BU17} \\
     \hline
     $O(\alpha\cdot \frac{\log \Delta}{\log\log \Delta}) $ & $O( \frac{\log \Delta}{\log\log \Delta})$ & \checkmark & $\times$ & \makecell[c]{\textbf{Our}\\\textbf{Result}} \\
     \hline
    \end{tabular}
\end{table}
\renewcommand{\arraystretch}{1}